\documentclass[pdflatex,sn-mathphys-num]{sn-jnl}

\usepackage{graphicx}%
\usepackage{multirow}%
\usepackage{amsmath,amssymb,amsfonts}%
\usepackage{amsthm}%
\usepackage{mathrsfs}%
\usepackage[title]{appendix}%
\usepackage{xcolor}%
\usepackage{textcomp}%
\usepackage{manyfoot}%
\usepackage{booktabs}%
\usepackage{algorithm}%
\usepackage{algorithmicx}%
\usepackage{algpseudocode}%
\usepackage{listings}%
\usepackage{yfonts}%

\theoremstyle{thmstyleone}%
\newtheorem{theorem}{Theorem}
\newtheorem{proposition}[theorem]{Proposition}%

\theoremstyle{thmstyletwo}%
\newtheorem{example}{Example}%

\theoremstyle{thmstylethree}%
\newtheorem{definition}{Definition}%

\def\1{\mathbf{1}}

\def\cC{\mathscr{C}}

\def\cN{\mathscr{N}}

\def\CC{{\mathbb C}}
\def\RR{{\mathbb R}}

\def\EE{{\mathbb E}}

\def \0{{\mathbf{0}}}

\def \vecw {\mathbf{w}}
\def \vecx {\mathbf{x}}
\def \vecy {\mathbf{y}}
\def \vecz {\mathbf{z}}

\def \vecc {\mathbf{c}}
\def \vecx {\mathbf{x}}

\makeatletter
\newcommand{\superimpose}[2]{%
	{\ooalign{$#1\@firstoftwo#2$\cr\hfil$#1\@secondoftwo#2$\hfil\cr}}}
\makeatother
\newcommand{\intC}{\mathpalette\superimpose{{\hspace{0.3mm}\vspace{-0.2mm}\textsc{C}}{\int}}}

\begin{document}

\title[Properties of the core and other solution concepts of Bel coalitional
  games in the ex-ante scenario]{Properties of the core and other solution concepts of Bel coalitional games in the ex-ante scenario}


\author[1]{\fnm{Michel} \sur{Grabisch}}\email{michel.grabisch@univ-paris1.fr}

\author*[2]{\fnm{Silvia} \sur{Lorenzini}}\email{silvia.lorenzini@dottorandi.unipg.it}
\equalcont{These authors contributed equally to this work.}

\affil[1]{\orgdiv{Charles University}, \orgaddress{\state{Prague}, \country{Czech Republic}}}
\affil[1]{\orgdiv{Paris School of Economics}, \orgname{Universit\'e Paris I Panth\'eon-Sorbonne, Centre d'Economie de la Sorbonne}, \orgaddress{\state{Paris}, \country{France}}}

\affil*[2]{\orgdiv{Dipartimento di Economia}, \orgname{Universit\`a degli Studi di Perugia}, \orgaddress{ \state{Perugia}, \country{Italy}}}


\abstract{We study the properties of the core and other solution concepts of Bel
  coalitional games, that generalize classical coalitional games by introducing
  uncertainty in the framework. In this uncertain environment, we work with
  contracts, that specify how agents divide the values of the coalitions in the
  different states of the world. Every agent can have different a priori
  knowledge on the true state of the world, which is modeled through the
  Dempster-Shafer theory, while agents' preferences between contracts are
  modeled by the Choquet integral. We focus on the “ex-ante” scenario, when the
  contract is evaluated before uncertainty is resolved. We investigate the
  geometrical structure of the ex-ante core when agents have the same a priori
  knowledge which is a probability distribution. Finally, we define the
  (pre)nucleolus, the kernel and the bargaining set ({\it \`a la} Mas-Colell) in the
  ex-ante situation and we study their properties. It is found that the
  inclusion relations among these solution concepts are the same as in the
  classical case. Coincidence of the ex-ante core and the ex-ante bargaining set
  holds for convex Bel coalitional games, at the price of strengthening the
  definition of bargaining sets.

}





\maketitle


\section{Introduction}\label{sec1}
Cooperative games model situations where players cooperate in order to increase
their benefit or diminish the cost of using a service. In cooperative games with
transferable utility, it is assumed that money is transferable between players
by side payments, and that utilities are linear and increasing in money, with
the same slope for every player. As a consequence, a cooperative game with
transferable utility (TU-game) can be represented by a set function $v$,
assigning a real number $v(S)$ to every coalition $S$, called its worth, and
quantifying the benefit earned by coalition $S$  obtained through cooperation.

Most of the literature considers that $v(S)$ is perfectly known and is a
single real number. However, uncertainty or imprecision may occur in many
practical situations, making the use of the classical framework difficult or
unrealistic. There is a rather limited literature on cooperative games with
imprecision (interval-valued, or fuzzy-valued, see, e.g.,
\cite{brbralti10,mares94,weber10}), as well as a literature on games with uncertainty. Our
contribution falls into the latter category. Games with uncertainty seem to have
been first proposed by Charnes and Granot \cite{Charnes76} already fifty years
ago. Subsequent notable works have been achieved by Habis and Herings
\cite{hahe10,hahe11,hahe11a}, and Suijs et al. \cite{SUIJS99}, among others. Our
own proposition, called Bel coalitional games \cite{GL2025}, is inspired by the
concept of Bayesian coalitional games, proposed by Ieong and Shoham \cite{IS}. In
the latter approach, the value $v(S)$ depends on the state of nature, supposed to be
unknown, and players have a common prior probability distribution on the states
of nature. Then, payoff vectors are replaced by contracts, i.e., agreements
between players which stipulate the payoff of each player under each state of
nature. Ieong and Shoham distinguish between the ex-ante situation, where no
information on the true state of nature is revealed, the interim situation,
where some information is revealed, and the ex-post situation where uncertainty
is resolved. In particular they study the core of such games in the different
situations. As further explained below, our approach is similar but differ in
some respects: we do not suppose that all players have the same prior, and we do
not suppose it is a probability distribution, but rather a belief function,
which are uncertainty measures more general than probability measures,
introduced by Dempster \cite{D} and Shafer \cite{shafer}.

In \cite{GL2025}, we have studied in detail the core of Bel coalitional games in
the ex-ante and the interim situation (there is not much to say in the ex-post
situation, as this is the classical case). The present paper is a follow-up of
this work, focusing on the ex-ante situation, which is the most interesting. Our
aim is to study the geometrical structure of the ex-ante core, to propose
ex-ante versions of other classical solution concepts and to explore
their properties. Concerning the ex-ante core, its study in the general case (no
common prior, belief functions) revealed to be difficult and it was not possible
to find a necessary and sufficient condition for a contract to belong to the
core. However, such a condition is available when there is a common prior which
is a probability distribution (the Bayesian case). Our first main result in the
present paper is to give, thanks to this condition, a complete description of
the ex-ante core. It is found that it is an affine space, whose pseudo-vertices are
those of the core of the ``expected'' game, i.e., the classical cooperative game
$V$ where $V(S)$ is the expected value of $v_\omega(S)$, the uncertain game. As
for other solution concepts, we propose a definition of the bargaining set ({\it
  \`a la} Mas-Colell \cite{MC1989}), as well as definitions of the (pre)nucleolus and the
kernel. Our second main achievement is to show that the classical inclusion
relations between these three solution concepts persist in the uncertain
case. Finally, we study the case of convex Bel coalitional games, introduced in
\cite{GL2025}, and show that the ex-ante bargaining set coincides with the core
for convex games, at the price of strengthening its definition: the ex-ante
strong bargaining set can be seen as an intermediate between the Mas-Colell
definition and the original Davis and Maschler definition \cite{dama67}.

The paper is organized as follows. Section~\ref{sec:prel} introduces the
necessary material (Dempster-Shafer theory, cooperative games) in order to
introduce Bel coalitional games and the main results obtained for the ex-ante
core from \cite{GL2025} in Section~\ref{sec:bel}. We also introduce the Shapley
contract, which is simply the contract giving the Shapley value of the game in
each state of nature. Section~\ref{sec:struc}
studies the geometry of the ex-ante core. In Section~\ref{sec:barg}, we
introduce the ex-ante bargaining set, the ex-ante (pre)nucleolus and the ex-ante
kernel, and study their inclusion properties. In Section~\ref{sec:conv}, we
study the properties of ex-ante convex Bel coalitional games, and show that the
Shapley contract lies in the ex-ante core of the game, and that the ex-ante
strong bargaining set coincides with the ex-ante core for such games.

\section{Preliminaries}\label{sec:prel}
\subsection{Dempster-Shafer Theory}
Let $\Omega = \{\omega_1,\omega_2,\ldots,\omega_d\}$ be a finite non-empty set of states of the world. We denote by $2^\Omega$ and $\mathbb{R}^\Omega$ the power set of $\Omega$ and the set of all random variables on $\Omega$, respectively.  A \textit{capacity} is a set function $\nu:2^\Omega\rightarrow [0,1]$ satisfying:
\begin{itemize}
\item[\it (i)] $\nu(\emptyset)= 0$ and $\nu(\Omega)= 1$;

\item[\it (ii)] $\nu(A)\le\nu(B)$, whenever $A \subseteq B$.
\end{itemize}
A \textit{belief function} (see \cite{D,shafer}) is a particular capacity satisfying:
\begin{equation}\label{eq:tm}
\nu\left(\bigcup^k_{i=1} E_i\right)\ge \sum_{\emptyset \ne I \subseteq
  \{1,\ldots,k\}} (-1)^{|I|+1} \nu\left(\bigcap_{i\in I} E_i\right),
\end{equation}
for all $ k\ge 2$ and $E_1,\ldots,E_k\in 2^\Omega$.
Notice that a probability measure $\pi$ on $2^\Omega$ is a particular belief function where (\ref{eq:tm}) holds with equality.

A belief function $\nu$ is associated with a dual set function $\overline{\nu}$ on $2^\Omega$ called \textit{plausibility function} and defined, for all $A\in 2^\Omega$, as $\overline{\nu}(A) =1-\nu(A^c).$
Both belief and plausibility functions are completely characterized by their
\textit{M\"obius inverse} (or \textit{mass function}) $m_\nu: 2^\Omega
\rightarrow [0,1]$ via the formulas
\[
\nu(A)=\sum_{B\subseteq A} m_\nu(B)\quad\mbox{and}\quad
\overline{\nu}(A)=\sum_{B\cap A \neq \emptyset} m_\nu(B).
\]
A subset of $\Omega$ with strictly positive mass function is called a {\it focal element} and we denote by $\mathcal{F}_{\nu} = \{E \in 2^\Omega \,:\, m_\nu(E) > 0\}$ the set of focal elements of $\nu$.

Given a belief function $\nu$ and $X \in \mathbb{R}^\Omega$, the {\it Choquet expectation} of $X$ with respect to $\nu$ is defined through the Choquet
integral (see, e.g., \cite{MG}):
$$\mathbb{C}_{\nu}(X):=\intC X \,d\nu =\sum^d_{i=1}[X(\omega_{\sigma(i)})-X(\omega_{\sigma(i+1)})]\nu(E_i^\sigma),$$ 
where $\sigma$ is a permutation of $\{1,\ldots,d\}$ such that
$X(\omega_{\sigma(1)})\ge \cdots \ge X(\omega_{\sigma(d)})$,
$E_i^\sigma=\{\omega_{\sigma(1)},\ldots,\omega_{\sigma(i)}\}$, for
$i=1,\ldots,d$, and $X(\omega_{\sigma(d+1)})=0$. When $\nu$ reduces to a
probability measure $\pi$, we have that $\mathbb{C}_\pi(X)=\mathbb{E}_\pi(X)$,
the expected value of $X$.

Let us summarize some useful properties of the Choquet integral with respect to a belief function $\nu$:
\begin{itemize}
\item \textit{Positive homogeneity} (it holds also for $\overline{\nu}$):
  $\displaystyle\intC \alpha X \,d\nu =\alpha \intC X \,d\nu$, with $\alpha\ge
  0$;
\item \textit{Asymmetry}: $\displaystyle\intC -X \,d\nu=- \intC X \,d\overline{\nu}$;
\item \textit{Monotonicity} (or \textit{nondecreasingness}) with respect to the integrand (it holds also for $\overline{\nu}$): $$X\le Y \implies \displaystyle\intC X \,d\nu\le \intC Y \,d\nu;$$
\item \textit{Subadditivity}: $\displaystyle\intC X \,d\overline{\nu}+\intC Y \,d\overline{\nu}\ge \intC X+Y \,d\overline{\nu}$:
\item \textit{Superadditivity}: $\displaystyle\intC X \,d\nu+\intC Y \,d\nu\le \intC X+Y \,d\nu$; 
\item {\it Continuity:} $\displaystyle\intC Xd\nu$ is a
  continuous function w.r.t. $X$.
\end{itemize}

\subsection{Coalitional Games}\label{sec:coga}
The aim of the theory of coalitional games (or cooperative games, TU-games) is the study of payoff division
within groups of agents. A game assigns to each group of agents, called a
coalition, a payoff (or utility), representing benefit or cost saving resulting
from the cooperation of the agents in that coalition.
\begin{definition}
    A \textit{coalitional game} is a pair $(N,v)$ where
    \begin{itemize}
        \item $N= \{1,2,\ldots,n\}$ is a set of agents;
        \item $v:2^N\rightarrow\mathbb{R}$ is a function that assigns to each
          group of agents $S\subseteq N$ a real-valued payoff, and $v(\emptyset)=0$.
    \end{itemize}
\end{definition}
If the context makes $N$ clear, we can refer to the coalitional game simply as $v$. An outcome or payoff vector $\textbf{x}$ is a vector in $\mathbb{R}^N$ defining the payoff given to any agent in $N$.

\begin{definition}\label{def:core}
The \textit{core} of a coalitional game $(N,v)$ is defined by 
    $$C(N,v)=\{\textbf{x}\in\mathbb{R}^N:\sum_{i\in S} \textbf{x}_i \ge v(S),\,\forall\, S\subseteq N,\,\sum_{i\in N} \textbf{x}_i = v(N)\}.$$
\end{definition}
For a given game, the core may
be empty, i.e., there may be no payoff vector that satisfies the stated
condition. 



\section{Bel coalitional games}\label{sec:bel}
In this section we present coalitional games with uncertainty and Bel coalitional games, the latter introduced in \cite{GL2025}.
\subsection{Coalitional games with uncertainty}
Up to our knowledge, there is not a large literature dealing
  with cooperative games incorporating uncertainty or imprecision. The simplest
  models consider that $v(S)$ is an imprecise quantity, either modeled by an
  interval \cite{brbralti10}, an ellipsoid \cite{weber10}, or a fuzzy number
  \cite{mares94}. Considering $v(S)$ as a random variable seems to have been
  first introduced by Charnes and Granot \cite{Charnes76}. A more involved
  framework has been proposed later by Suijs et al. \cite{SUIJS99}, where
  players have a set of possible actions, giving rise to a payoff. More
  recently, in a series of papers, Habis and Herings consider cooperative games
  as random variables, with players having continuous and state-separable utility
  functions \cite{hahe10,hahe11,hahe11a}. Their study concentrates on the weak
  sequential core, a notion originally proposed by Kranich et al. \cite{krpepe05}.

Perhaps the closest work to ours is the one of Ieong and Shoham
  \cite{IS}, who propose Bayesian coalitional games, whose definition is given
  below.

\begin{definition}
    A \textit{Bayesian coalitional game} is given by $(N,\Omega,\pi,(\mathcal{I}_j)_{j\in N},(\succcurlyeq_j)_{j\in N})$ where
    \begin{itemize}
        \item $N= \{1,2,\ldots,n\}$ is a set of agents;
        \item $\Omega= \{\omega_1,\omega_2,\ldots,\omega_d\}$ is a set of possible worlds, where each world specifies a coalitional game defined over $N$;
        \item $\pi$ is a common prior probability distribution over the set of
          worlds $\Omega$;
        \item $\mathcal{I}_j$ is agent $j$'s partition of the worlds in $\Omega$
          (information set);
        \item $\succcurlyeq_j$ describes agent $j$'s preference over distributions of payoffs, for each agent $j$.
    \end{itemize}
\end{definition}
The common prior probability distribution on the states of
  nature is the feature which distinguishes Bayesian coalitional games from
  games with uncertainty of Halbis and Herings. The information sets of the
  players are publicly known, and represent the way the players perceive
  uncertainty: two states of nature in a block of the partition $\mathcal{I}_j$
  are indistinguishable by player $j$. In addition, Ieong and Shoham considers
  three different situations: ex-ante, interim and ex-post. In the ex-ante
  situation, players have no information about the true state of nature, and
  their preference on uncertain payoffs (called contracts) is governed by
  expected utility. In the interim situation, to each player is revealed in
  which block of their information set the true state of nature lies. The
  resulting notion of core (set of unblocked payoffs) is quite complex because
  of the public character of the information. In the
  ex-ante situation, uncertainty is resolved by revealing the true state of
  nature, and we are back to classical cooperative games.

\subsection{Bel coalitional games}
Bel coalitional games, proposed by the authors in \cite{GL2025}, are largely
inspired by Bayesian coalitional games, although some aspects are simplified,
while others are generalized. Their formal definition is given below.
\begin{definition}
    A \textit{Bel coalitional game} is given by $(N,\Omega,(m_{\nu_j})_{j\in N},(v_\omega)_{\omega\in\Omega})$ where
    \begin{itemize}
        \item $N= \{1,2,\ldots,n\}$ is a set of agents;
        \item $\Omega= \{\omega_1,\omega_2,\ldots,\omega_d\}$ is a set of
          possible worlds, where each world $\omega$ specifies a coalitional game $v_\omega$ defined over $N$;
         \item $m_{\nu_j}:2^\Omega\rightarrow [0,1]$ is agents $j$'s mass distribution over the worlds in $\Omega$, and $\nu_j$ is the corresponding belief function, such that $\bigcup\mathcal{F}_{\nu_j}=\Omega$ (the
             focal elements cover $\Omega$);
       \item $v_\omega: 2^N \rightarrow \mathbb{R}$ is a coalitional game that assigns to each coalition $S\subseteq N$ its value in world $\omega$.
    \end{itemize}
\end{definition}

We do not assume any more that agents have the same prior, and we use the richer
notion of belief functions to model uncertainty. This implies that the Choquet
integral has to be used for computing expected values of uncertain payoffs
(contracts). 
On the other hand, we do not
consider information sets any more, meaning that agents perceive all states of
nature. 

Although in \cite{GL2025} we consider the three situations
  (ex-ante, interim, ex-post), in this paper we focus on the ex-ante
  situation. The reason is that, unlike for Bayesian games, the interim
  situation does not differ so much in terms of results, and the ex-post
  situation coincides with classical cooperative games.

\medskip

In the ex-ante situation, as the worth of a game is uncertain,
  one has to consider as well uncertain payoffs. They specify
  for each state of nature how the worth of a coalition should be divided. Borrowing the terminology of
  Bayesian games, we call them {\it contracts}, showing that these are
  agreements between players of a coalition, bound before uncertainty is resolved.

\begin{definition}
A {\it contract} among agents of coalition $S\subseteq N$ ($S$-contract) is a
mapping from the set of worlds to payoff vectors, $\textbf{c}^S :
\Omega\rightarrow\mathbb{R}^S$, such that $\textbf{c}_j^S(\omega)$ denotes the
payoff to agent $j\in S$ in world $\omega$. An $S$-contract is {\it feasible} if for all worlds $\omega$ $$\sum_{j\in S} \textbf{c}_j^S(\omega) \le v_\omega(S).$$
\end{definition}
 For any individual agent $j$, contracts play no role and agent $j$ simply receives
its own payoff $v_\omega(\{j\})$ in world $\omega$. As in \cite{IS}, we focus only on feasible
contracts and we study the stability of $N$-contracts (grand contracts). We say that a (feasible) grand contract $\textbf{c}^N$ is {\it efficient at $\omega$} if $\sum_{j\in
  N}\textbf{c}^N_j(\omega)=v_\omega(N)$. Therefore, it is not efficient if $\sum_{j\in
  N}\textbf{c}^N_j(\omega)<v_\omega(N)$.
We say that $\textbf{c}^N$ is {\it efficient} if it is efficient at every
$\omega\in\Omega$.

The next step is to define preference over contracts. Let
  $\vecc^S,\vecc'^S$ be two contracts for $S$. We have that $\vecc^S\succcurlyeq_S
  \vecc'^S$ if every agent $j$ in $S$ prefers the payoff distribution $\vecc^S_j$ to $\vecc'^S_j$, which we
  denote by $\vecc^S_j\succcurlyeq_j\vecc'^S_j$. We write  $\vecc^S\succ_S
  \vecc'^S$ if the preference is strict for each agent in $S$.  As explained above,
  we consider that preference over payoff distributions is given by
  the Choquet integral, i.e., for any agent $j\in S$: 

\[
\vecc^S_j\succcurlyeq_j\vecc'^S_j \Leftrightarrow \mathbb{C}_{\nu_j}(\textbf{c}_j^{S})\ge\mathbb{C}_{\nu_j}({\textbf{c}'}_j^{S}).
\]

\begin{example}
 Classical point-valued solution concepts can be used to
    define particular grand contracts, for example the Shapley value \cite{sha53}.  Given a
    Bel coalitional game $(N,\Omega,(m_{\nu_j})_{j\in
      N},(v_\omega)_{\omega\in\Omega})$, we define the {\it Shapley contract}
    as $$\textbf{c}^N_{Sh}(\omega)=\phi(v_\omega),$$ where
$$\phi_i(v_\omega) = \sum_{S \subseteq N \setminus \{i\}} \frac{|S|! (n - |S| - 1)!}{n!} \left( v_\omega (S \cup \{i\}) - v_\omega(S) \right).$$

The Shapley contract inherits the properties of the classical Shapley value,
    most notably efficiency (see definition above). Other classical properties
    are the null player property and symmetry. Adapting them to our framework,
    we may say that a player $i$ is {\it null} if for any state $\omega$, $i$ is
    null in $v_\omega$, i.e., $v_\omega(S\cup\{i\})=v_\omega(S)$, for all
    $S\subseteq N\setminus \{i\}$. Similarly, players $i$ and $j$ are {\it
      symmetric} if for any state $\omega$ we have
    $v_\omega(S\cup\{i\})=v(S\cup\{j\})$, for all $S\subseteq
    N\setminus\{i,j\}$. Then the Shapley contract satisfies the null property
    (for any null player $i$, $(\vecc^N_{Sh})_i(\omega)=0$ for all
    $\omega\in\Omega$) and the symmetry property (if $i,j$ are symmetric
    players,  then $(\textbf{c}^N_{Sh})_i(\omega)=(\textbf{c}^N_{Sh})_j(\omega)$
    for every $\omega\in\Omega$). Finally, the Shapley contract can be
    characterized by efficiency, the null player property, symmetry, and
    additivity (defined statewise in the obvious way). 
\end{example}

We introduce the notion of blocking of a grand contract by a coalition $S$.

\begin{definition}
Given a Bel coalitional game $(N,\Omega,(m_{\nu_j})_{j\in
  N},(v_\omega)_{\omega\in\Omega})$ and a grand contract $\textbf{c}^N$, a
coalition $S$ {\it ex-ante blocks} $\textbf{c}^N$ if there exists an $S$-contract $\textbf{c}^S$ such that
$$\textbf{c}^S\succ_S\textbf{c}_S^N,$$
where with $\textbf{c}_S^N$ we mean $(\textbf{c}^N_j)_{j\in S}$.
\end{definition}
The ex-ante core of Bel coalitional games is defined using the notion of blocking.

\begin{definition}
The {\it ex-ante core} of a Bel coalitional game is the set of all feasible grand contracts $\textbf{c}^N$ for which no coalition $S\subseteq N$ ex-ante blocks $\textbf{c}^N$.
\end{definition}

The properties of the ex-ante core of Bel coalitional games have been studied at length in
  \cite{GL2025}, in particular the conditions for nonemptiness. Its study reveals to be difficult in the general case, when
  considering different priors, which are belief functions, and no
  characterization result has been obtained. When all priors are identical, the
  following sufficient condition has been obtained and will be used in
  Section~\ref{sec:conv}.
  \begin{theorem}\label{th:conv}
    Given a Bel coalitional game $(N,\Omega,m_\nu,(v_\omega)_{\omega\in\Omega})$
    with all identical priors, and a grand contract $\vecc^N$, if for all
    $S\subseteq N$ it holds
    \[
\CC_\nu(v_\omega(S))\leq \sum_{i\in S}\CC_\nu(\vecc^N_i),
\]
then the grand contract $\vecc^N$ is in the ex-ante-core of the Bel coalitional game.
    \end{theorem}
Now, if the common prior is a probability distribution, the above condition
becomes necessary and sufficient (also proved in \cite{IS}).  

\begin{theorem}\label{thprob}
Given a Bel coalitional game $(N,\Omega,m_{\nu},(v_\omega)_{\omega\in\Omega})$, with $\nu=\pi$, then a grand contract $\textbf{c}^N$ is in the ex-ante core of the game if and only if, for all $S\subseteq N$
\begin{equation}\label{eq:exac}
  \mathbb{E}_{\pi}(v_\omega(S))\le\mathbb{E}_{\pi}(\sum_{i\in S}\textbf{c}^N_i).
\end{equation}
\end{theorem}

\section{Structure of the ex-ante core}\label{sec:struc}
In order to study the geometric structure of the ex-ante core, it is necessary
to have a characterization result, i.e., an if-and-only-if condition for a
grand contract to belong to the core. As explained above, this has been obtained
only if  one supposes the same prior for
all agents, which has to be a probability measure. Therefore, throughout this
section we consider a Bel coalitional game
$(N,\Omega,m_\nu,(v_\omega)_{\omega\in\Omega})$ with $\nu=\pi$, where $\pi$ is a
probability measure, and we suppose that this game is balanced, i.e., it has a
nonempty ex-ante core. According to Theorem~\ref{thprob} above, we know that a
grand contract $\textbf{c}^N$ is in the ex-ante core iff (\ref{eq:exac}) is satisfied.

The aim of this section is to study the geometrical structure of the ex-ante
core. We introduce
some notation: for a (classical) TU-game $v$ on $N$, its core is denoted by $C(v)$ and for a Bel coalitional game, with  $\nu_j=\pi$ for all
$j\in N$, we denote the ex-ante core as $C((v_\omega)_{\omega\in \Omega},\pi)$.

By definition of Bel coalitional games, $\pi$ has full support $\Omega$, i.e.,
there is no $\omega\in \Omega$ such that $\pi(\omega)=0$. 

In this section, in order to simplify the notation and to make it closer to the
case of classical TU games, instead of denoting a grand contract by $\vecc^N$ as
a vector-valued function on $\Omega$, we denote it by $\vecx$, a vector in
$\RR^{N\times \Omega}$, with coordinates $\textbf{x}^\omega_i$,
$i\in N$, $\omega\in \Omega$.

Recall that any grand contract $\vecx$ satisfies feasibility:
\[
\sum_{i\in N}\textbf{x}^\omega_i\leq v_\omega(N), \quad (\omega\in\Omega).
\]
If $\textbf{x}\in C((v_\omega)_{\omega\in\Omega},\pi)$, it satisfies Equation (\ref{eq:exac}):
\[
\EE(\sum_{i\in N}\textbf{x}^\omega_i)\geq \EE(v_\omega(N)).
\]
Since $\pi$ has full support, this forces
\[
\sum_{i\in N}\textbf{x}^\omega_i= v_\omega(N), \quad (\omega\in\Omega).
\]
Therefore, letting $\Omega=\{\omega_1,\ldots,\omega_d\}$ and denoting for
simplicity $\pi(\omega_j)=:\pi_j$, $\textbf{x}^{\omega_j}_i=:\textbf{x}^j_i$, and
$v_{\omega_j}=v_j$, the ex-ante core is
the set of vectors $\textbf{x}\in \RR^{N\times\Omega}$ satisfying the following set of
inequalities:
\begin{align}
  \sum_{j=1}^d\pi_j\sum_{i\in S}\textbf{x}^j_i &\geq \sum_{j=1}^d\pi_jv_j(S), \quad
  S\subset N\label{eq:c1}\\
  \sum_{i\in N}\textbf{x}^j_i & = v_j(N), \quad j\in \{1,\ldots,d\}\label{eq:c2}
\end{align}
We may use the shorthand $\sum_{i\in S}\textbf{x}^j_i=:\textbf{x}^j(S)$ to simplify the
notation. Also, for $\textbf{x}\in\RR^{N\times\Omega}$, $\textbf{x}^j$ is the vector in $\RR^N$
with coordinates $\textbf{x}^j_i$, $i\in N$. 
We see that the ex-ante core is a convex polyhedron, whose structure
remains to be elucidated.

Some simple results can be established. Let us introduce the expected TU-game
$V$ on $N$ defined by
\[
V(S) =\sum_{\omega\in\Omega}\pi(\omega)v_{\omega}(S), \quad (S\subseteq N).
\]
\begin{proposition}\label{prop:sr}
  The following holds.
  \begin{enumerate}
  \item Let us consider $\textbf{x}^j\in C(v_j)$ for $j=\{1,\ldots,d\}$. Then
    $\sum_{j=1}^d\pi_j\vecx^j$ is an element of $C(V)$.
  \item If $\textbf{x}\in\RR^{N\times\Omega}$ is an ex-ante core element then
    $\textbf{y}=\sum_{j=1}^d\pi_j\textbf{x}^j\in\mathbb{R}^N$ belongs to $C(V)$.
    \item If $\textbf{y}\in C(V)$, then any solution
      $(\textbf{x}^1,\ldots,\textbf{x}^d)\in\RR^{N\times\Omega}$ to the following system
      \begin{align}
        \sum_{j=1}^d\pi_j\textbf{x}^j &= \textbf{y}\label{eq:c3}\\
        \textbf{x}^j(N) & = v_j(N), \quad j=1,\ldots,d\label{eq:c4}
      \end{align}
       is an element of the ex-ante core.
  \end{enumerate}
\end{proposition}
\begin{proof}
  (i) As $\textbf{x}^j$ satisfies $\textbf{x}^j(N)=v_j(N)$ and $\textbf{x}^j(S)\geq v_j(S)$ for all
  $S\subseteq N$ and $j=1,\ldots,d$, the result follows.

  (ii) Take $\textbf{x}\in C((v_\omega)_{\omega\in \Omega},\pi)$ and put $\textbf{y}:=\sum_{j=1}^d\pi_j\textbf{x}^j$. Since $\textbf{x}$ satisfies
  (\ref{eq:c2}), this implies $\textbf{y}(N)=\sum_{j=1}^d\pi_j\textbf{x}^j(N)=\sum_{j=1}^d\pi_jv_j(N)=V(N)$. Since $\textbf{x}$
  satisfies (\ref{eq:c1}), it follows
  $\textbf{y}(S)=\sum_{j=1}^d\pi_j\textbf{x}^j(S)\geq\sum_{j=1}^d\pi_jv_j(S)=V(S)$. Therefore
  $\textbf{y}\in C(V)$.

  (iii) Take $\textbf{y}\in C(V)$ and consider any solution
  $\textbf{x}=(\textbf{x}^1,\ldots,\textbf{x}^d)\in\RR^{N\times\Omega}$ to (\ref{eq:c3}) and
  (\ref{eq:c4}). Then $\textbf{x}$ satisfies (\ref{eq:c2}). Also,
  $V(S)=\sum_{j=1}^d\pi_jv_j(S)\leq \textbf{y}(S)=\sum_{j=1}^d\pi_j\textbf{x}^j(S)$ by
  (\ref{eq:c3}), hence $\textbf{x}$ satisfies also (\ref{eq:c1}) and is therefore an element of the ex-ante core. 
  \end{proof}

Let us note that (ii) and (iii) characterize the ex-ante core, and may indicate a practical way
to find all elements in the ex-ante core. From (iii), we see that the ex-ante
core is unbounded, since if $\textbf{x}^j_i$ tends to $+\infty$, it suffices that another
coordinate $\textbf{x}^j_{i'}$ tends to $-\infty$.

\begin{proposition}\label{prop:ir}
Consider a balanced Bel coalitional game $(N,\Omega,m_\nu,(v_\omega)_{\omega\in\Omega})$
with $\nu=\pi$. Then, its ex-ante core $C((v_\omega)_{\omega\in\Omega},\pi)$
  \begin{enumerate}
    \item 
 contains a lineality space of dimension
$(n-1)(d-1)$. A basis for this space is
$(\vecw^1_2,\ldots,\vecw^1_n,\vecw^2_2,\ldots,\vecw^2_n,\ldots, \vecw^{d-1}_2,\ldots,\vecw^{d-1}_n)$,
given by
\begin{align*}
  \vecw^1_2 &=
  \Big(\underbrace{-1,1,0,\ldots,0}_{j=1},0\ldots,0,\underbrace{\frac{\pi_1}{\pi_d},-\frac{\pi_1}{\pi_d},0,\ldots,0}_{j=d}\Big)\\
  \vecw^1_n &=
  \Big(\underbrace{-1,0,\ldots,0,1}_{j=1},0\ldots,0,\underbrace{\frac{\pi_1}{\pi_d},0,\ldots,0,-\frac{\pi_1}{\pi_d}}_{j=d}\Big)\\
    \vecw^2_2 &=
  \Big(0,\ldots,0,\underbrace{-1,1,0,\ldots,0}_{j=2},0\ldots,0,\underbrace{\frac{\pi_2}{\pi_d},-\frac{\pi_2}{\pi_d},0,\ldots,0}_{j=d}\Big)\\
  \vecw^2_n &=
  \Big(0,\ldots,0,\underbrace{-1,0,\ldots,0,1}_{j=2},0\ldots,0,\underbrace{\frac{\pi_2}{\pi_d},0,\ldots,0,-\frac{\pi_2}{\pi_d}}_{j=d}\Big)\\
  \vecw^{d-1}_2 &=
  \Big(0,\ldots,0,\underbrace{-1,1,0,\ldots,0}_{j=d-1},\underbrace{\frac{\pi_{d-1}}{\pi_d},-\frac{\pi_{d-1}}{\pi_d},0,\ldots,0}_{j=d}\Big)\\
  \vecw^{d-1}_n &=
  \Big(0,\ldots,0,\underbrace{-1,0,\ldots,0,1}_{j=d-1},\underbrace{\frac{\pi_{d-1}}{\pi_d},0,\ldots,0,-\frac{\pi_{d-1}}{\pi_d}}_{j=d}\Big)
\end{align*}
\item has no other extremal ray;
\item has no vertices.
  \end{enumerate}
  \end{proposition}
  \begin{proof}
    (i)
Recall that the lineality space of a polyhedron defined by $A\textbf{x}\geq b$ is given
by the set of solutions of $A\textbf{x}=0$. In our case, this amounts to solving the
system
\begin{align}
  \sum_{j=1}^d\pi_j\sum_{i\in S}\textbf{x}^j_i &=0, \quad
  S\subset N\label{eq:p1}\\
  \sum_{i\in N}\textbf{x}^j_i & = 0, \quad j\in \{1,\ldots,d\}.\label{eq:p2}
\end{align}
Considering the equations pertaining to $|S|=1$, we get:
\begin{equation}\label{eq:p3}
\textbf{x}^d_i = -\frac{1}{\pi_d}(\pi_1\textbf{x}^1_i+ \cdots + \pi_{d-1}\textbf{x}^{d-1}_i) \quad (i\in N).
\end{equation}
Observe that any equation pertaining to any $S$ with $|S|>1$ is redundant, since
it can be obtained as the sum of the equations $\sum_{j=1}^d\pi_j\textbf{x}^j_i=0$ for
all $i\in S$. It remains the equations $\sum_{i\in N}\textbf{x}^j_i=0$, $j=1,\ldots,d$,
from which we get
\[
\textbf{x}_1^j = -\textbf{x}_2^j- \cdots-\textbf{x}_n^j \quad (j=1,\ldots,d).
\]
Observe however that the last one ($j=d$) is redundant, since it can be obtained
from summing all equations of (\ref{eq:p3}). Therefore, we have eliminated
$n+d-1$ variables, which makes the dimension of the linear subspace equal to
$(n-1)(d-1)$. The general form of a vector $\vecw$ in the lineality space is:
\begin{multline*}
\vecw=
\Bigg(\underbrace{-\sum_{i=2}^n\textbf{x}_i^1,\textbf{x}_2^1,\ldots,\textbf{x}_n^1}_{j=1},\ldots,\underbrace{-\sum_{i=2}^n\textbf{x}_i^{d-1},\textbf{x}_2^{d-1},\ldots,\textbf{x}_n^{d-1}}_{j=d-1},\\ \underbrace{\frac{1}{\pi_d}\Big(\pi_1\sum_{i=2}^n\textbf{x}_i^1+\cdots+\pi_{d-1}\sum_{i=2}^n\textbf{x}_i^{d-1}\Big),
-\frac{1}{\pi_d}\sum_{i=1}^{d-1}\pi_i\textbf{x}^i_2,\ldots,-\frac{1}{\pi_d}\sum_{i=1}^{d-1}\pi_i\textbf{x}^i_n}_{j=d}\Bigg) 
\end{multline*}
with $\textbf{x}^1_2,\ldots,\textbf{x}^1_n,\textbf{x}^2_2,\ldots,\textbf{x}^2_n,\ldots,
\textbf{x}^{d-1}_2,\ldots,\textbf{x}^{d-1}_n\in \RR$.
The expression of the basis follows from this general form.

(ii) The conic part (recession cone) of the ex-ante core is given by the system:
\begin{align}
  \sum_{j=1}^d\pi_j\sum_{i\in S}\textbf{x}^j_i &\geq 0, \quad
  S\subset N\label{eq:c11}\\
  \sum_{i\in N}\textbf{x}^j_i & = 0, \quad j\in \{1,\ldots,d\}.\label{eq:c22}
\end{align}
Observe that by (\ref{eq:c22}) we have
\[
0=\sum_{j=1}^d\pi_j\sum_{i\in N}\textbf{x}^j_i =\sum_{i\in N}\sum_{j=1}^d\pi_j\textbf{x}^j_i.
\]
By (\ref{eq:c11}) for $|S|=1$, each term $\sum_{j=1}^d\pi_j\textbf{x}^j_i$ is nonnegative, therefore
each term must be zero:
\[
\sum_{j=1}^d\pi_j\textbf{x}^j_i=0\quad (i\in N)
\]
from which we deduce that equality holds in each equation of
(\ref{eq:c11}). Consequently, the recession cone amounts to be identical to the
lineality space.

(iii) By (i), the ex-ante core has no vertices since it contains a lineality space.
  \end{proof}

Recall that by the Weyl-Minkovski theorem any convex closed $n$-dimensional
polyhedron $P$ (whose representation by inequalities is $A\vecx \geq b$) can be
decomposed into its convex and its conic parts:
\[
P = \mathrm{conv}(\vecx^1 , \ldots , \vecx^q ) + \mathrm{cone}(\vecz^1 , \ldots, \vecz^r ).
\]
The above theorem has found the conic part (recession cone) of the ex-ante core,
it remains to find the convex part. When the recession cone is pointed, $\vecx^1
, \ldots, \vecx^q$ correspond to the vertices of $P$ , i.e., they satisfy a
subsystem of $m$ independent equalities in $A\vecx \geq b$, where $m$ is the
dimension of the polyhedron. Otherwise, there is no vertex in this sense, but
``pseudo-vertices'' which are solution of a subsystem of $m-l$ independent
equalities, where $l$ is the dimension of the lineality space. If $\vecx$ is a
pseudo-vertex, then $\vecx + \vecw$ is another one, where $\vecw$ belongs to the
lineality space.

Applying these facts to our case permits to get a complete description of the
ex-ante core.

\begin{proposition}\label{prop:mr}
Consider a balanced Bel coalitional game $(N,\Omega,m_\nu,(v_\omega)_{\omega\in\Omega})$
with $\nu=\pi$. The following holds.
\begin{enumerate}
\item $C((v_\omega)_{\omega\in\Omega},\pi)\neq\emptyset$ iff
  $C(V)\neq\emptyset$.
\item Supposing $C(V ) \neq \emptyset$, to any vertex $\vecy \in \RR^N$ of $C(V
  )$ corresponds a pseudo-vertex $\vecx \in \RR^{N \times\Omega}$ of
  $C((v_\omega)_{\omega\in\Omega},\pi)$ and conversely, where the correspondance
  is given by (\ref{eq:c3}) and (\ref{eq:c4}).
\end{enumerate}
\end{proposition}
\begin{proof}
1. Suppose the ex-ante core is nonempty. Then there exists
$\vecx\in\RR^{N\times\Omega}$, $\vecx = (\vecx^1 ,
\ldots, \vecx^d )$, satisfying the system (\ref{eq:c1}),
(\ref{eq:c2}). Letting $\vecy :=\sum_{j=1}^d\pi_j\vecx^j$, by Proposition~\ref{prop:sr} (2), this
vector belongs to $C(V)$. Conversely, supposing $C(V)$ is nonempty, we can pick a
core element $\vecy$. Then any solution of (\ref{eq:c3}), (\ref{eq:c4}) is an
element of the ex-ante core by Proposition~\ref{prop:sr} (3).

2. Take a vertex $\vecy$ of $C(V)$, and consider a solution $\vecx = (\vecx^1 ,
\ldots, \vecx^d )\in \RR^{N\times\Omega}$ of (\ref{eq:c3}), (\ref{eq:c4}). We
know already by Proposition~\ref{prop:sr} (3) that $\vecx$ is an element of the
ex-ante core. From Proposition~\ref{prop:ir} (1), we have to prove that this ex-ante
core element must satisfy a subsystem of $nd-(n-1)(d-1) = n + d - 1$ independent
equations. Observe that (\ref{eq:c2}) provides $d$ independent equalities for
$\vecx$. Since $\vecy$ is a vertex of the core, it satisfies a subsystem of at
least $n - 1$ independent equations $\vecy(S) = V(S)$, $S \in \cC$, with
$N\not\in \cC$, and $\cC\subseteq 2^N$. Therefore, by (\ref{eq:c3}), we have
\[
\sum_{i\in S}\sum_{j=1}^d\pi_j\vecx_i^j = \sum_{j=1}^d\pi_jv_j(S),\ \forall S\in\cC.
\]
Hence, $n-1$ independent equations in (\ref{eq:c1}) are satisfied, and as a conclusion, $\vecx$ is a pseudo-vertex of the ex-ante core.

Conversely, suppose that $\vecx = (\vecx^1 , \ldots, \vecx^d )$ is a pseudo-vertex of the
ex-ante core. Consider $\vecy := \sum_{j=1}^d \pi_j \vecx^j$. By
Proposition~\ref{prop:sr} (2), $\vecy$ is an element of $C(V)$. As $\vecx$ is a
pseudo-vertex, it satisfies a subsystem of $n - 1$ independent equalities in the
system (\ref{eq:c1}), which translates into a subsystem of $n - 1$ independent
equalities $\vecy(S) = V (S)$. Therefore, $\vecy$ is a vertex of $C(V)$.
\end{proof}
  
\begin{example}
  Take $n=3$, $d=2$ and the uniform distribution for $\pi$. Consider $v_\omega$
 a symmetric game (i.e., $v_\omega(S)$ depends only on the cardinality of $S$) defined by
 \begin{center}
   \begin{tabular}{|c|c|c|c|}\hline
     & $|S|=1$ & $|S|=2$ & $N$\\ \hline
     $v_1$ & 2 & 6 & 10 \\
     $v_2$ & 0 & 4 & 8\\
     $V$ & 1 & 5 & 9\\ \hline
    \end{tabular}
 \end{center}
 where $V=\frac{1}{2}(v_1+v_2)$. Observe that all three games are convex. The
 vertices of $C(V)$ are the 3 marginal vectors
 \[
\textbf{y}^1=(1 \ 4 \ 4), \ \textbf{y}^2=(4\ 1\ 4), \ \textbf{y}^3=(4\ 4\ 1).
\]
The ex-ante core is described by the following system:
\begin{align*}
\textbf{x}^1_1+\textbf{x}^2_1 & \geq 2\\
\textbf{x}^1_2+\textbf{x}^2_2 & \geq 2\\
\textbf{x}^1_3+\textbf{x}^2_3 & \geq 2\\
\textbf{x}^1_1+\textbf{x}^1_2+\textbf{x}^2_1+\textbf{x}^2_2 &\geq 10\\
\textbf{x}^1_1+\textbf{x}^1_3+\textbf{x}^2_1+\textbf{x}^2_3 &\geq 10\\
\textbf{x}^1_2+\textbf{x}^1_3+\textbf{x}^2_2+\textbf{x}^2_3 &\geq 10\\
\textbf{x}^1_1+\textbf{x}^1_2+\textbf{x}^1_3 & = 10\\
\textbf{x}^2_1+\textbf{x}^2_2+\textbf{x}^2_3 & = 8\\
\end{align*}
Let  us find the conic part, and first check if the cone is pointed by finding
its lineality space. The lineality space is the set of solutions of the system:
\begin{align*}
\textbf{x}^1_1+\textbf{x}^2_1 & = 0\\
\textbf{x}^1_2+\textbf{x}^2_2 & = 0\\
\textbf{x}^1_3+\textbf{x}^2_3 & = 0\\
\textbf{x}^1_1+\textbf{x}^1_2+\textbf{x}^2_1+\textbf{x}^2_2 &= 0\\
\textbf{x}^1_1+\textbf{x}^1_3+\textbf{x}^2_1+\textbf{x}^2_3 &= 0\\
\textbf{x}^1_2+\textbf{x}^1_3+\textbf{x}^2_2+\textbf{x}^2_3 &= 0\\
\textbf{x}^1_1+\textbf{x}^1_2+\textbf{x}^1_3 & = 0\\
\textbf{x}^2_1+\textbf{x}^2_2+\textbf{x}^2_3 & = 0\\
\end{align*}
We find that the set of solutions is a linear subspace of the form
$$\{(-\alpha-\beta,\alpha,\beta,\alpha+\beta,-\alpha,-\beta),\;\alpha,\beta\in
\RR\},$$
hence generated by the 2 vectors $(-1,1, 0,1,-1,0)$ and
$(-1,0,1,1,0,-1)$. This is in accordance with the results given in
Proposition~\ref{prop:ir}.

Let us find the pseudo-vertices, using Proposition~\ref{prop:mr}. Each
pseudo-vertex corresponds to one of the three vertices of $C(V)$ given
above. Taking $\vecy^1$, the corresponding pseudo-vertex $\vecx=(\vecx^1,\vecx^2)$ is a solution
of the system:
\begin{align*}
  \vecx^1_1+\vecx^2_1 & = 2\\
  \vecx^1_2+\vecx^2_2 & = 8\\
  \vecx^1_3+\vecx^2_3 & = 8\\
  \vecx^1_1+\vecx^1_2+\vecx^1_3 & = 10\\
  \vecx^2_1+\vecx^2_2+\vecx^2_3 & = 8
\end{align*}
A solution of this system has the form
\[
(\vecx^1_1,\vecx^1_2,\vecx^1_3,\vecx^2_1,\vecx^2_2,\vecx^2_3)=(\alpha,\beta,10-\alpha-\beta,2-\alpha,8-\beta,\alpha+\beta-2),
\]
for any $\alpha,\beta\in\RR$. Proceeding similarly with $\vecy^1,\vecy^2$ we
find respectively
\[
(\alpha,\beta,10-\alpha-\beta,8-\alpha,2-\beta,\alpha+\beta-2), \ (\alpha,\beta,10-\alpha-\beta,8-\alpha,8-\beta,\alpha+\beta-8).
\]
A verification using the PORTA (POlyhedron Representation Transformation
Algorithm) software\footnote{\tt https://porta.zib.de} yields the above vectors
for the lineality space and the following pseudo-vertices:
\[
(2,4,4,0,4,4),\ (2,4,4,6,-2,4), \ (2,4,4,6,4,-2).
\]
This corresponds to the three pseudo-vertices above with $\alpha=2$ and
$\beta=4$. 

In summary, the ex-ante core is the convex hull of 3 affine subspaces, parallel
to the subspace generated by $(-1,1, 0,1,-1,0)$ and
$(-1,0,1,1,0,-1)$, each containing one pseudo-vertex.
\hfill$\diamond$
\end{example}
 
\section{The ex-ante bargaining set and related solution concepts}\label{sec:barg}

In Bel coalitional games, following \cite{IS}, we assume that agents enter into
contracts, that are exogenously given, about how to divide the values of the
coalitions. In order to take into account negotiations between coalitions,
assuming that agents are allowed to propose contracts and to
  bargain, we propose the notion of ex-ante bargaining set, in the spirit of the
  classical definitions of the literature on coalitional games.

There are several definitions of the bargaining set. Among the
    most common ones we find the definition of Davis and Maschler \cite{dama67},
    and the one of Mas-Colell \cite{MC1989}. In the former definition, given a
    payment vector $x$, a player $k$ addresses an objection to another player
    $l$ by exhibiting a payment $y$ for a coalition $C$ containing $k$ but not
    $l$, which is affordable and strictly better than $x$ for all members of
    $C$. The objection is justified if $l$ cannot build an objection against
    $k$, and the bargaining set is the set of imputations which have no
    justified objection.  In the definition of Mas-Colell (modified later by
    Shimomura \cite{shi97}), the difference is that, instead of one player
    having objection against another player, we have that one coalition has an
    objection against another one. It appears that the definition of an
    objection at $x$ by a coalition $S$ is close to our notion of blocking,
    where a coalition $S$ blocks some grand contract by an $S$-contract. For
    this reason we adopt the definition of Mas-Colell. It is also to be noted
    that Dutta et al. \cite{dutta1989} have further refined the bargaining set
    of Mas-Colell, defining it in the general framework of NTU-games, and
    proposing the notion of {\it consistent bargaining set}. In their
    definition, an objection is valid (or justified) if no {\it justified}
    counterobjection exists (i.e., which has no justified
    counter-counterobjection, and so on, till exhaustion).

The main objective of this section is to introduce the
    ex-ante bargaining set, as well as related solution concepts, namely the
    ex-ante (pre)nucleolus and the ex-ante kernel. We study their
    interrelationships and obtain that thay are nonempty.  

\subsection{The ex-ante bargaining set}

\begin{definition}
Given a Bel coalitional game $(N,\Omega,(m_{\nu_j})_{j\in
  N},(v_\omega)_{\omega\in\Omega})$, a grand contract $\textbf{c}^N$ and an
$S$-contract $\textbf{c}^S$ such that $\textbf{c}^S\succ_S\textbf{c}_S^N$n
(i.e., $S$ blocks $\vecc^N$),
a coalition $S'$ {\it ex-ante counterblocks} $\textbf{c}^S$ if there exists a feasible $S'$-contract $\textbf{c}^{S'}$ such that
\begin{itemize}
\item $\mathbb{C}_{\nu_j}(\textbf{c}^{S'}_j)\ge\mathbb{C}_{\nu_j}(\textbf{c}^S_j)$, for all $j\in S'\cap S$
\item $\mathbb{C}_{\nu_j}(\textbf{c}^{S'}_j)\ge\mathbb{C}_{\nu_j}(\textbf{c}_j^N)$, for all $j\in S'\setminus S$
\end{itemize}
and at least one of these inequalities is strict.
\end{definition}

Hence, we can say that a coalition $S$ {\it legitimate ex-ante blocks} a grand
contract $\textbf{c}^N$ if there does not exist any counterblocking coalition
$S'$ to the contract $\textbf{c}^{S}$. Now we can finally give the definition of
ex-ante bargaining set, using the notion of legitimate blocking.

\begin{definition}
The {\it ex-ante bargaining set} of a Bel coalitional game is the set of all feasible grand contracts $\textbf{c}^N$ for which no coalition $S\subseteq N$ legitimate ex-ante blocks $\textbf{c}^N$.
\end{definition}

From the definitions it is clear that, in general, the ex-ante core is included
in the ex-ante bargaining set, as for classical coalitional games.

\subsection{The ex-ante nucleolus and the ex-ante kernel}
For simplicity let us assume in this section that agents have same knowledge
given by $\nu$. We introduce the notions of prenucleolus, nucleolus \cite{Schmeidler1969}
and kernel \cite{DM1963} in our
framework, and study inclusion properties among these notions.

  First of all, we define the {\it ex-ante excess} of some coalition $S$
to a grand contract $\textbf{c}^N$ as
$$e(S,\textbf{c}^N)=\mathbb{C}_{\nu}(v_\omega(S))-\sum_{i\in S}\mathbb{C}_{\nu}(\textbf{c}^N_i).$$
The ex-ante excess measures how dissatisfied the members of the coalition $S$ are with the grand contract $\textbf{c}^N$, before the ambiguity about the state of the world is resolved.
Then we introduce the vector of decreasing excesses
$$\theta(\textbf{c}^N)=(e(S_1,\textbf{c}^N),\dots,e(S_{2^N},\textbf{c}^N))$$
where $S_i$, $i=1,\dots,2^N$, are all the coalitions indexed such that
$$e(S_1,\textbf{c}^N)\ge\cdots\ge e(S_{2^N},\textbf{c}^N).$$ To minimize the
excess we compare any two vectors $\theta(\textbf{c}^N)$ and
$\theta({\textbf{c}'}^N)$ by the lexicographic order, denoted by
$\succsim_L$. Given two vectors $\theta(\textbf{c}^N)$ and
$\theta({\textbf{c}'}^N)$, we write
$\theta(\textbf{c}^N)\succsim_L\theta({\textbf{c}'}^N)$ if either
$\theta(\textbf{c}^N)=\theta({\textbf{c}'}^N)$ or there exists $k$,
$1\le k\le N$, such that $\theta_k(\textbf{c}^N)>\theta_k({\textbf{c}'}^N)$ and
$\theta_i(\textbf{c}^N)=\theta_i({\textbf{c}'}^N)$, for all $i$ such that $1\le
i<k$. In particular we have that
$\theta(\textbf{c}^N)\succ_L\theta({\textbf{c}'}^N)$ if
$\theta(\textbf{c}^N)\succsim_L \theta({\textbf{c}'}^N)$ and
$\theta({\textbf{c}'}^N)\not\succsim_L\theta({\textbf{c}}^N)$.

We give first the definition of prenucleolus and nucleolus. 
\begin{definition}
Let $X\subset \RR^{N\times\Omega}$. The {\it ex-ante nucleolus w.r.t. $X$} of a
Bel coalitional game is the set of all feasible grand contracts $\textbf{c}^N$
such that $\theta(\textbf{c}^N)\precsim_L\theta({\textbf{c}'}^N)$, for all
$\vecc'^N\in X$.
\end{definition}

The prenucleolus is defined as the nucleolus w.r.t. the set of all efficient grand
contracts. This is similar to the case of classical cooperative games.
\begin{definition}
The {\it ex-ante prenucleolus} of a Bel coalitional game is the set of all
efficient grand contracts $\textbf{c}^N$ such that
$\theta(\textbf{c}^N)\precsim_L\theta({\textbf{c}'}^N)$, for all efficient ${\textbf{c}'}^N$.
\end{definition}

To prove the nonemptiness of the prenucleolus, we use the expected game $V$
defined by
\[
V(S) = \CC_\nu(v_\omega(S)),\ S\subseteq N,
\]
which is a classical cooperative game. Given a payoff vector $\vecx\in\RR^N$,
recall that the excess of $\vecx$ at $S$ w.r.t. $V$ is defined by
$e(S,\vecx)=V(S)-\vecx(S)$, with $\vecx(S):=\sum_{i\in S}\vecx_i$. 

\begin{theorem}
For any Bel coalitional game, the prenucleolus is a nonempty set.
\end{theorem}
\begin{proof}
Consider a Bel coalitional game $(N,\Omega,m_\nu,(v_\omega)_{\omega\in\Omega})$,
and the expected game $(N,V)$ defined by
\[
V(S) = \CC_\nu(v_\omega(S)),\quad S\subseteq N.
\]
By standard results on cooperative games, the prenucleolus of $V$, denoted by
$\cN(V)$, exists and is a single point.

Let us show that the certain contract
$\tilde{\vecc}^N$  defined by
\[
\tilde{\vecc}^N_i(\omega) = \cN_i(V), \forall \omega\in\Omega,\forall i\in N
\]
belongs to the prenucleolus of the Bel coalitional game. Suppose not. Then there exists an efficient contract $\vecc^N$ such
that $\theta(\vecc^N)\prec_L\theta(\tilde{\vecc}^N)$.  Define
the following efficient payoff vector:
\[
\vecx_i=\CC_\nu(\vecc^N_i), \quad\forall
i\in N.
\]
Observe that for any $S\subseteq N$, the excesses of $\vecx$ w.r.t. $V$ satisfy
\[
e(S,\vecx)=V(S)-\vecx(S) = \CC_\nu(v_\omega(S))-\sum_{i\in
  S}\CC_\nu(\vecc^N_i)=e(S,\vecc^N),
\]
and similarly
\[
e(S,\tilde{\vecc}^N) = \CC_\nu(v_\omega(S))-\sum_{i\in
  S}\CC_\nu(\tilde{\vecc}^N_i)=V(S)-\sum_{i\in S}\cN_i(V)=e(S,\cN(V)).
\]
It follows that 
\[
\theta(\vecx)=\theta(\vecc^N)\prec_L\theta(\tilde{\vecc}^N)=\theta(\cN(V)),
\]
contradicting the fact that $\cN(V)$ is the prenucleolus of $V$.
\end{proof}

Observe however that, unlike the prenucleolus of classical cooperative games,
the ex-ante prenucleolus may not be a singleton and may not be a compact
set. Indeed, consider for example the case $\nu=\pi$, and take any efficient
grand contract $\vecc^N$ in the prenucleolus and define another efficient grand contract
$\vecc'^N$ , which differs from $\vecc^N$ inasmuch as
\[
\vecc'^N_i(\omega_1)=\vecc^N_i(\omega_1)+\frac{\delta}{\pi(\omega_1)}, \quad \vecc'^N_i(\omega_2)=\vecc^N_i(\omega_2)-\frac{\delta}{\pi(\omega_2)}
\]
for some $\delta\in\RR$, some $i\in N$ , and some $\omega_1,\omega_2\in\Omega$. Then $\vecc'^N$ belongs also to the
ex-ante prenucleolus, because $\EE_\pi (\vecc'^N_j)=\EE_\pi (\vecc^N_j )$ for all $j\in N$. As $\delta$ can be taken
arbitrarily large, the ex-ante prenucleolus is not a compact set.

\medskip

In order to define the ex-ante kernel we need the definition of ex-ante surplus. Let
$i,j$ be two distinct players in $N$. The \textit{ex-ante surplus} of $i$ against $j$ at $\textbf{c}^N$ is
$$s_{i,j}(\textbf{c}^N)=\max_{S\ni i,\,S\not\ni j}e(S,\textbf{c}^N).$$

\begin{definition}
The {\it ex-ante kernel} of a Bel coalitional game is the set of all feasible grand contracts $\textbf{c}^N$ such that for all $i,j\in N$
$$s_{i,j}(\textbf{c}^N)=s_{j,i}({\textbf{c}}^N).$$
\end{definition}

Proceeding like in \cite{Schmeidler1969}, we prove that the ex-ante prenucleolus is contained in the ex-ante kernel.

\begin{theorem}
The ex-ante prenucleolus of a Bel coalitional game is contained in the ex-ante kernel.
\end{theorem}

\begin{proof}
Let $\textbf{c}^N$ be in the ex-ante prenucleolus. By contradiction we assume that $\textbf{c}^N$ is not in the ex-ante kernel. Then, there exist $i,j\in N$ such that
$$s_{i,j}(\textbf{c}^N)>s_{j,i}({\textbf{c}}^N).$$
Now, let us define
$$\delta=s_{i,j}(\textbf{c}^N)-s_{j,i}({\textbf{c}}^N)>0$$
and consider the efficient grand contract $\tilde{\textbf{c}}^N$ such that
$$\tilde{\textbf{c}}^N_t(\omega)=\begin{cases}\textbf{c}^N_i(\omega)+\delta & \mbox{if}\,t=i\\
\textbf{c}^N_j(\omega)-\delta &\mbox{if}\,t=j\\
\textbf{c}^N_t(\omega) &\mbox{otherwise}\end{cases}$$
for all $\omega\in\Omega$.
We want to show that
$$\theta(\textbf{c}^N)\succ_L\theta({\tilde{\textbf{c}}}^N).$$

We define the set $K=\{S\subseteq N \,:\, e(S,\textbf{c}^N)\ge s_{i,j}(\textbf{c}^N)\,\mbox{and}\,i\in S\implies j\in S\}$ and we assume that $|K|=k\,\in\mathbb{N}$. By construction, we have that $\theta_k(\textbf{c}^N)\ge s_{i,j}(\textbf{c}^N)$, but $\theta_{k+1}(\textbf{c}^N)=s_{i,j}(\textbf{c}^N)$, since there is at least one coalition $R$ such that $i\in R$, $j\not\in R$ and $e(R,\textbf{c}^N)= s_{i,j}(\textbf{c}^N)$. For such coalition, using the translation invariance of Choquet integral, results that
\begin{eqnarray*}
e(R,\tilde{\textbf{c}}^N)&=&\mathbb{C}_{\nu}(v_\omega(R))-\sum_{t\in R}\mathbb{C}_{\nu}(\tilde{\textbf{c}}^N_t)\\
&=&\mathbb{C}_{\nu}(v_\omega(R))-[\sum_{t\in R\setminus\{i\}}\mathbb{C}_{\nu}(\textbf{c}^N_t)+\mathbb{C}_{\nu}(\textbf{c}^N_i+\delta)]\\
&=&\mathbb{C}_{\nu}(v_\omega(R))-[\sum_{t\in R\setminus\{i\}}\mathbb{C}_{\nu}(\textbf{c}^N_t)+\mathbb{C}_{\nu}(\textbf{c}^N_i)+\delta]\\
&=&\mathbb{C}_{\nu}(v_\omega(R))-\sum_{t\in R}\mathbb{C}_{\nu}(\textbf{c}^N_t)-\delta\\
&=&s_{i,j}(\textbf{c}^N)-\delta\\
&<&s_{i,j}(\textbf{c}^N)=e(R,\textbf{c}^N).
\end{eqnarray*}
Therefore, $\theta_{l}(\textbf{c}^N)=\theta_{l}(\tilde{\textbf{c}}^N)$, for $l=1,\dots, k$, and $\theta_{k+1}(\textbf{c}^N)>\theta_{k+1}(\tilde{\textbf{c}}^N)$, that is a contradiction.
\end{proof}

To conclude, we have the following statement.

\begin{theorem}
The ex-ante kernel of a Bel coalitional game is contained in the ex-ante bargaining set.
\end{theorem}

\begin{proof}
Let $\textbf{c}^N$ be in the ex-ante kernel. If $\textbf{c}^N$ is in the ex-ante core then it is in the ex-ante bargaining set. If $\textbf{c}^N$ is not in the ex-ante core then there exists a coalition $\overline{S}$ that ex-ante blocks $\textbf{c}^N$. 
Now, let us consider $i\in \overline{S}$ and $j\not\in \overline{S}$ and a coalition $T$ such that $j\in T$, $i\not\in T$ and
$$s_{j,i}(\textbf{c}^N)=e(T,\textbf{c}^N).$$
Since $\textbf{c}^N$ is in the ex-ante kernel we have that
\begin{eqnarray*}
\sigma :=e(T,\textbf{c}^N)&=&s_{j,i}(\textbf{c}^N)\\
&=&s_{i,j}(\textbf{c}^N)\\
&\ge & e(\overline{S},\textbf{c}^N)\\
&=&\mathbb{C}_{\nu}(v_\omega(\overline{S}))-\sum_{k\in\overline{S}}\mathbb{C}_{\nu}(\textbf{c}^{N}_k)\\
&\ge &\mathbb{C}_{\nu}(\sum_{k\in\overline{S}}(\textbf{c}^{\overline{S}}_k))-\sum_{k\in\overline{S}}\mathbb{C}_{\nu}(\textbf{c}^{N}_k)\\
&\ge &\sum_{k\in\overline{S}}\mathbb{C}_{\nu}(\textbf{c}^{\overline{S}}_k)-\sum_{k\in\overline{S}}\mathbb{C}_{\nu}(\textbf{c}^{N}_k)\\
&>& \sum_{k\in T\cap\overline{S}}\mathbb{C}_{\nu}(\textbf{c}^{\overline{S}}_k)-\sum_{k\in T\cap\overline{S}}\mathbb{C}_{\nu}(\textbf{c}^{N}_k)=:\alpha,
\end{eqnarray*}
where we use the feasibility of the $\overline{S}$-contract, the superadditivity
of the Choquet integral, and the fact that $\overline{S}$ blocks $\vecc^N$. Hence $\delta :=\sigma-\alpha>0$.
Let us consider the following $T$-contract for every $\omega\in\Omega$:
\begin{eqnarray*}
\textbf{c}^T_k(\omega) & = & \mathbb{C}_{\nu}(\textbf{c}^{\overline{S}}_k)+\frac{\delta}{|T|}-\epsilon,\;\forall k\in T\cap\overline{S}\\
\textbf{c}^T_k(\omega)& = & \mathbb{C}_{\nu}(\textbf{c}^N_k)+\frac{\delta}{|T|}+\frac{|T\cap \overline{S}|\epsilon+v_\omega(T)-\mathbb{C}_{\nu}(v_\omega(T))}{|T\setminus \overline{S}|},\;\forall k\in T\setminus\overline{S}
\end{eqnarray*}

where $0<\epsilon\leq \frac{\delta}{|T|}$.

The contract $\textbf{c}^T$ is feasible, since for all $\omega\in\Omega$
\begin{eqnarray*}
\sum_{k\in T}\textbf{c}^T_k(\omega) & = & \sum_{k\in T\cap \overline{S}} \mathbb{C}_{\nu}(\textbf{c}^{\overline{S}}_k)+|T\cap \overline{S}|\frac{\delta}{|T|}-|T\cap \overline{S}|\epsilon+\sum_{k\in T\setminus\overline{S}}\mathbb{C}_{\nu}(\textbf{c}^N_k)+|T\setminus\overline{S}|\frac{\delta}{|T|}+|T\cap \overline{S}|\epsilon+v_\omega(T)-\mathbb{C}_{\nu}(v_\omega(T))\\
&=&\sum_{k\in T\cap \overline{S}} \mathbb{C}_{\nu}(\textbf{c}^{\overline{S}}_k)+\sum_{k\in T\setminus\overline{S}}\mathbb{C}_{\nu}(\textbf{c}^N_k)+\delta+v_\omega(T)-\mathbb{C}_{\nu}(v_\omega(T))\\
&=&\mathbb{C}_{\nu}(v_\omega(T))+v_\omega(T)-\mathbb{C}_{\nu}(v_\omega(T))=v_\omega(T)
\end{eqnarray*}
and, by construction, we have

\begin{align*}
  \mathbb{C}_{\nu}(\textbf{c}^{T}_k) & = \CC_\nu(\vecc^{\overline{S}}_k) +
  \frac{\delta}{|T|} - \epsilon
  \ge\mathbb{C}_{\nu}(\textbf{c}^{\overline{S}}_k), \quad \forall k\in T\cap \overline{S}\\
  \mathbb{C}_{\nu}(\textbf{c}^{T}_k) & = \CC_\nu(\vecc^N_k)+\frac{\delta}{|T|} + \frac{|T\cap\overline{S}|}{|T\setminus\overline{S}|}\epsilon  >\mathbb{C}_{\nu}(\textbf{c}_k^N),\quad \forall k\in T\setminus \overline{S}. 
\end{align*}

Then the coalition $T$ ex-ante counterblocks $\textbf{c}^{\overline{S}}$, proving that $\textbf{c}^N$ is in the ex-ante bargaining set.
\end{proof}

We can draw two important conclusions of these series of results: First, the
ex-ante bargaining set, the prenucleolus and the kernel always exist. Second,
they may not be compact sets, unlike the classical case. Note that this is also the
case for the ex-ante core.

\section{Convex Bel coalitional games}\label{sec:conv}
The convexity property of Bel coalitional games has been defined in \cite{GL2025}. We
say that
$(N,\Omega,m_{\nu},(v_\omega)_{\omega\in\Omega})$ is {\it ex-ante convex} if for all $S,T\subseteq N$
$$\mathbb{C}_{\overline{\nu}}(v_\omega(S))+\mathbb{C}_{\overline{\nu}}(v_\omega(T))\le
\mathbb{C}_{\nu}(v_\omega(S\cup T))+\mathbb{C}_{\nu}(v_\omega(S\cap T)).$$

First, we show that when the Bel coalitional game is ex-ante convex, the Shapley contract is in the ex-ante core.

\begin{theorem}
Let us consider a Bel coalitional game $(N,\Omega,m_{\nu},(v_\omega)_{\omega\in\Omega})$. If it is ex-ante convex, then the Shapley contract is in the ex-ante core of the game.
\end{theorem}

\begin{proof}
Let us note that the Shapley value of the game $v_\omega$ can be written as
$$\phi(v_\omega) = \frac{1}{n!}\sum_{\sigma\in\textswab{G}(N)} m^{\sigma,v_\omega},$$
where $m^{\sigma,v_\omega}$ denotes the marginal vector of $v_\omega$ for permutation $\sigma$. In Theorem 12 of \cite{GL2025} we prove that $m^{\sigma,v_\omega}$ is in the ex-ante core of the game for all $\sigma\in\textswab{G}(N)$. Now, let us fix a $S\subseteq N$, we have that
\begin{eqnarray*}
\sum_{i\in S}\mathbb{C}_{\nu}(\textbf{c}^N_{Sh,i})&=&\sum_{i\in S}\mathbb{C}_{\nu}(\frac{1}{n!} \sum_{\sigma\in\textswab{G}(N)} m^{\sigma,v_\omega}_i)\\
&=&\frac{1}{n!}\sum_{i\in S}\mathbb{C}_{\nu}(\sum_{\sigma\in\textswab{G}(N)} m^{\sigma,v_\omega}_i)\\
&\ge & \frac{1}{n!}\sum_{i\in S}\sum_{\sigma\in\textswab{G}(N)}\mathbb{C}_{\nu}(m^{\sigma,v_\omega}_i)\\
&=&\frac{1}{n!}\sum_{\sigma\in\textswab{G}(N)}\sum_{i\in S}\mathbb{C}_{\nu}(m^{\sigma,v_\omega}_i)\\
&\ge &\frac{1}{n!}\sum_{\sigma\in\textswab{G}(N)}\mathbb{C}_{\nu}(v_\omega(S))\\
&=&\frac{1}{n!}n!\,\mathbb{C}_{\nu}(v_\omega(S))=\mathbb{C}_{\nu}(v_\omega(S)),
\end{eqnarray*}
where, in order, we use the asymmetry and the superadditivity property of the
Choquet integral, we change the order of summation and finally we use the fact
that $m^{\sigma,v_\omega}$ is in the ex-ante core of the game for all
$\sigma\in\textswab{G}(N)$ and Theorem~\ref{th:conv}. We conclude from the above
inequality and Theorem~\ref{th:conv} again that the Shapley contract belongs to
the ex-ante core.
\end{proof}

\medskip

It is well-known that for classical cooperative games, the bargaining set in the
sense of Davis and Maschler coincides with the core for convex games. Up to our
knowedge, this property has not been proved for the bargaining set in the sense
of Mas-Colell. Dutta et al. \cite{dutta1989} have shown this property for their
notion of consistent bargaining set in the framework of NTU games, where
``consistent'' means that a counterobjection must be valid, in the sense that it
cannot be itself counterobjected, and so on. As a result, their consistent bargaining set
is  included in the bargaining set of Mas-Colell, as it is more difficult to
counterobject objections.

What we propose below is also a strengthening of our original definition,
resulting in a smaller bargaining set, which
makes true the coincidence of the core and the bargaining set for convex
games. In a sense, it is closer to the definition of Davis and Maschler, as we
will explain.

\begin{definition}
The {\it ex-ante strong bargaining set} of a Bel coalitional game is the set of all feasible grand contracts $\textbf{c}^N$ for which every ex-ante blocking coalition $S\subseteq N$ admits a counterblocking coalition $S'$ to the contract $\textbf{c}^{S}$ such that $S'\not\subseteq S$.
\end{definition}

Recall that in the definition of Davis and Maschler, objections
  are defined w.r.t. players: $i$ has an objection against $j$ if there is a
  coalition $S$ containing $i$ but not $j$ and an imputation on $S$ making all
  players in $S$ better off. Then a counterobjection of $j$ against $i$ amounts
  to find a coalition $T$ containing $j$ but not $i$, and an imputation better
  for players in $T$. As a consequence, it can never be the case that $T$ is
  included in $S$. This is exactly what we propose: the coalition which makes a
  counterobjection cannot be a subset of the coalition making an objection.
  
\begin{theorem}\label{thcorebarg}
Let us consider a Bel coalitional game $(N,\Omega,m_{\nu},(v_\omega)_{\omega\in\Omega})$, with $\nu=\pi$. If it is ex-ante convex, then the ex-ante strong bargaining set coincide with the ex-ante core.
\end{theorem}

\begin{proof}
By definition, we know that the ex-ante core is included in ex-ante strong bargaining set,
and we therefore need to show that the ex-ante core contains the ex-ante strong bargaining
set. By contradiction, let us assume that there exists a feasible grand contract
$\textbf{c}^N$ in the ex-ante strong bargaining set that is not in the ex-ante
core. Since $\textbf{c}^N$ is fixed throughout the proof, we denote the ex-ante
excess of the coalition $S$ at $\textbf{c}^N$ by $e(S):=e(S,\textbf{c}^N)$. As
this defines a mapping from $2^N$ to $\RR$, we may consider $e$ as a (classical)
coalitional game on
$N$, which we denote by $(N,e)$. We claim that
$(N,e)$ is convex. Indeed, for every $A,B\subseteq N$,
\begin{eqnarray*}
e(A)+e(B)&=&\mathbb{C}_{\pi}(v_\omega(A))-\sum_{i\in A}\mathbb{C}_{\pi}(\textbf{c}^N_{i})+\mathbb{C}_{\pi}(v_\omega(B))-\sum_{i\in B}\mathbb{C}_{\pi}(\textbf{c}^N_{i})\\
&= &\mathbb{E}_{\pi}(v_\omega(A))+\mathbb{E}_{\pi}(v_\omega(B))-\sum_{i\in A\cup B}\mathbb{E}_{\pi}(\textbf{c}^N_{i})-\sum_{i\in A\cap B}\mathbb{E}_{\pi}(\textbf{c}^N_{i})\\
&\le & \mathbb{E}_{\pi}(v_\omega(A\cup B))+\mathbb{E}_{\pi}(v_\omega(A\cap B))-\sum_{i\in A\cup B}\mathbb{E}_{\pi}(\textbf{c}^N_{i})-\sum_{i\in A\cap B}\mathbb{E}_{\pi}(\textbf{c}^N_{i})\\
&=&e(A\cup B)+e(A\cap B),
\end{eqnarray*}
where we have used the ex-ante convexity of the Bel coalitional game. Next, we
consider the coalitional game $(N,\hat{e})$ that is the monotonic cover of the
game $(N,e)$, i.e.,
$$\hat{e}(S)=\max_{R\subseteq S} e(R).$$
 $(N,\hat{e})$ is a convex game, since the monotonic cover of any
convex game is a convex game (see, e.g., \cite{ZMS}). Now, from among all the
coalitions that have maximal ex-ante excess for $\textbf{c}^N$, we choose $S^*$
which is maximal w.r.t. set inclusion, i.e.,
\begin{eqnarray*}
e(S)& \le & e(S^*),\; \forall S\subseteq N\\
e(S)& < & e(S^*),\; \forall S: S^*\subset S\subseteq N
\end{eqnarray*}
Since $\textbf{c}^N$ is not in the ex-ante core, there exists a coalition
$\overline{S}$ such that 
$$e(\overline{S})=\mathbb{E}_{\pi}(v_\omega(\overline{S}))-\sum_{i\in \overline{S}}\mathbb{E}_{\pi}(\textbf{c}^N_i)>0$$
By the fact that $e(S^*)$ is maximal we have that
\begin{equation*}
\hat{e}(S^*)=\max_{R\subseteq S^*} e(R)= e(S^*)\geq e(\overline{S})>0
\end{equation*}
Then, let us consider the subgame $(S^*,\hat{e})$ of the coalitional game
$(N,\hat{e})$, that is convex as a subgame of a convex game. Hence, the core of
$(S^*,\hat{e})$ is nonempty, and there exists a payoff vector $\textbf{x}$ in
the core of $(S^*,\hat{e})$, that is, satisfying
\begin{eqnarray}\label{eqx}
\sum_{i\in S^*}\textbf{x}_i & = & \hat{e}(S^*)=e(S^*)\nonumber\\
\sum_{i\in R}\textbf{x}_i & \ge & \hat{e}(R) \ge e(R),\;\forall R\subset S^*
\end{eqnarray}
In particular, for $R=\{i\}$, $\textbf{x}_i\geq\hat{e}(\{i\})\ge e(\emptyset)=0$,
i.e., every $\textbf{x}$ in the core of $(S^*,\hat{e})$ is a nonnegative vector. Since $\sum_{i\in S^*}\textbf{x}_i = \hat{e}(S^*)>0$, there is an agent $k\in S^*$ such that $\textbf{x}_k>0$. We define the following contract for every $\omega\in\Omega$:
\begin{eqnarray*}
\textbf{c}^{S^*}_j(\omega)& = &
\mathbb{E}_{\pi}(\textbf{c}^N_j)+\textbf{x}_j+\epsilon,\ \forall j\in S^*\setminus \{k\}\\
\textbf{c}^{S^*}_k(\omega) & = & \mathbb{E}_{\pi}(\textbf{c}^N_k)+\textbf{x}_k-\epsilon(|S^*|-1)+v_\omega(S^*)-\mathbb{E}_{\pi}(v_\omega(S^*)),
\end{eqnarray*}
with $\epsilon$ satisfying
\[
0<\epsilon<\min\Big(\frac{\vecx_k}{|S^*|-1},\epsilon'\Big)
\]
and
\[
\epsilon'=\min_{S:S\setminus S^*\neq\emptyset, S^*\setminus
  S\neq\emptyset}\Big(\frac{e(S^*)-e(S^*\cup S)}{|S^*\setminus S|}\Big),
\]
which is positive, assuming $S^*\neq N$ (when $S^*=N$, simply take
$\epsilon<\vecx_k/(|S^*|-1)$). 

1. Let us prove that $\vecc^{S^*}$ is feasible (and even efficient). We have:
\begin{eqnarray*}
\sum_{i\in S^*}\textbf{c}^{S^*}_i(\omega)&=&\sum_{i\in S^*}\mathbb{E}_{\pi}(\textbf{c}^N_i)+\sum_{i\in S^*}\textbf{x}_i+v_\omega(S^*)-\mathbb{E}_{\pi}(v_\omega(S^*))\\
&=&\sum_{i\in S^*}\mathbb{E}_{\pi}(\textbf{c}^N_i)+e(S^*)+v_\omega(S^*)-\mathbb{E}_{\pi}(v_\omega(S^*))\\
&=&\sum_{i\in S^*}\mathbb{E}_{\pi}(\textbf{c}^N_i)+\mathbb{E}_{\pi}(v_\omega(S^*))-\sum_{i\in S^*}\mathbb{E}_{\pi}(\textbf{c}^N_i)+v_\omega(S^*)-\mathbb{E}_{\pi}(v_\omega(S^*))\\
&=&v_\omega(S^*),
\end{eqnarray*}

2. Let us prove that $\textbf{c}^{S^*}$ ex-ante blocks $\textbf{c}_{S^*}^N$. We have: 
\begin{itemize}
\item for any $j\in S^*\setminus \{k\}$,
  \[
\mathbb{E}_\pi(\textbf{c}^{S^*}_j) = \EE_\pi(\vecc^N_j)+\vecx_j+\epsilon>\EE_\pi(\vecc^N_j).
  \]
\item for $k\in S^*$,
  \[
\mathbb{E}_\pi(\textbf{c}^{S^*}_k) =
\EE_\pi(\vecc^N_k)+\vecx_k-\epsilon(|S^*|-1)+\EE_\pi(v_\omega(S^*))-\EE_\pi(v_\omega(S^*))
>\EE_\pi(\vecc^N_k)
  \]
by assumption on $\epsilon$. 
Hence $\vecc^{S^*}$ ex-ante blocks $\textbf{c}^N$.
\end{itemize}

3. Since $\textbf{c}^N$ belongs to the ex-ante strong bargaining set, there must exist a counterblocking coalition $S'$ to the contract $\textbf{c}^{S^*}$ such that $S'\not\subseteq S^*$. Let us show
that this is not possible.

3.a. Assume $k\not\in S'\cap S^*$. We have:

\begin{eqnarray}
\sum_{i\in S'}\mathbb{E}_{\pi}(\textbf{c}^{S'}_i)&= &\mathbb{E}_{\pi}\Big(\sum_{i\in S'}\textbf{c}^{S'}_i\Big)\\ \label{eq1}
&\le & \mathbb{E}_{\pi}(v_\omega(S'))\\ \label{eq2}
&=& e(S')+\sum_{i\in S'}\mathbb{E}_{\pi}(\textbf{c}^N_i)\\ \label{eq3}
&\leq & e(S')+e(S^*)-e(S'\cup S^*)+\sum_{i\in S'}\mathbb{E}_{\pi}(\textbf{c}^N_i)\\ \label{eq4}
&\le & e(S'\cap S^*)+\sum_{i\in S'}\mathbb{E}_{\pi}(\textbf{c}^N_i)\\ \label{eq5}
&\leq&\sum_{i\in S'\cap S^*}\textbf{x}_i+\sum_{i\in S'}\mathbb{E}_{\pi}(\textbf{c}^N_i)\\ \label{eq6}
&=&\sum_{i\in S'\cap S^*}\vecc^{S^*}_i(\omega)-\sum_{i\in S'\cap
  S^*}\EE_\pi(\vecc^N_i) - |S'\cap S^*|\epsilon+\sum_{i\in
  S'}\EE_\pi(\vecc^N_i),\forall \omega\in\Omega\\ \label{eq7}
&=& \sum_{i\in S'\cap S^*}\EE_\pi(\vecc^{S^*}_i)+\sum_{i\in S'\setminus
  S^*}\EE_\pi(\vecc^N_i) - |S'\cap S^*|\epsilon\\ \label{eq8}
&\leq& \sum_{i\in S'\cap S^*}\EE_\pi(\vecc^{S^*}_i)+\sum_{i\in S'\setminus
  S^*}\EE_\pi(\vecc^N_i)
\end{eqnarray}
Equation~\eqref{eq1} follows from feasibility, Equation~\eqref{eq2} holds from
the definition of $e(S')$, Equation~\eqref{eq3} holds by definition of $S^*$ and
by the fact that $S'\cup S^*\supseteq S^*$, Equation~\eqref{eq4}
holds by convexity of $(S^*, e)$, Equation~\eqref{eq5} holds from
\eqref{eqx}, Equations~\eqref{eq6} and \eqref{eq7} hold from the definition of
$\textbf{c}^{S^*}$.
 It follows that $S^*$ legitimate ex-ante blocks
$\textbf{c}^N$ and this is a contradiction.

3.b. Assume that $k\in S'\cap S^*$. As above, we have
\[
\sum_{i\in
  S'}\mathbb{E}_{\pi}(\textbf{c}^{S'}_i)=\mathbb{E}_{\pi}\Big(\sum_{i\in
  S'}\textbf{c}^{S'}_i\Big) 
\le  \mathbb{E}_{\pi}(v_\omega(S'))
= e(S')+\sum_{i\in S'}\mathbb{E}_{\pi}(\textbf{c}^N_i).
\]
By definition of $\epsilon$, and since $S'\not\subseteq S^*$, we have:
\begin{eqnarray*}
  \sum_{i\in S'}\mathbb{E}_{\pi}(\textbf{c}^{S'}_i)+\epsilon|S^*\setminus S'|
  &\leq & e(S') + e(S^*) - e(S'\cup S^*) +\sum_{i\in
    S'}\mathbb{E}_{\pi}(\textbf{c}^N_i)\\
  &\leq & e(S'\cap S^*) + \sum_{i\in S'}\EE_\pi(\vecc^N_i)\\
  &\leq & \sum_{i\in S'\cap S^*}\vecx_i + \sum_{i\in S'}\EE_\pi(\vecc^N_i)\\
  &=& \sum_{i\in S'\cap S^*\setminus k}\vecc^{S^*}_i(\omega) - \sum_{i\in S'\cap
    S^*\setminus k}\EE_\pi(\vecc^N_i)-\epsilon(|S'\cap
  S^*|-1)+\vecc^{S^*}_k(\omega)\\ &&-\EE_\pi(\vecc^N_k)
  +\epsilon(|S^*|-1)-v_\omega(S^*)+\EE_\pi(v_\omega(S^*)) + \sum_{i\in S'}\EE_\pi(\vecc^N_i),
\end{eqnarray*}
for all $\omega\in \Omega$. Hence, taking expectation, we get
\[
\sum_{i\in S'}\mathbb{E}_{\pi}(\textbf{c}^{S'}_i)+\epsilon|S^*\setminus S'| \leq
\sum_{i\in S'\cap S^*}\EE_\pi(\vecc^{S^*}_i) + \sum_{i\in S'\setminus
  S^*}\EE_\pi(\vecc^N_i) + \epsilon|S^*\setminus S'|,
\]
hence,  $S^*$ legitimate ex-ante blocks $\textbf{c}^N$ , a contradiction.

\end{proof}

\section{Conclusion}
This article deepens the study of Bel coalitional games, by exploring the geometrical structure of the ex-ante core and the
properties of other solution concepts. We prove that
the ex-ante core is an affine space when there is a common prior which is a
probability distribution. Then, we define the ex-ante (pre)nucleolus, kernel and
bargaining set and we show that the relation between these three solutions is
the same as for classical coalitional games. Finally, we provide a stronger
definition of the ex-ante bargaining set and we show that it coincides with the
core for convex Bel coalitional games.

Much remains to be done for a complete exploration of the solution concepts of
Bel coalitional games. In particular, an interesting line of future research
concerns the study of the geometrical structure of the core in the general case
when agents have ambiguous knowledge modeled by a belief function. This means
also finding necessary and sufficient conditions for a grand contract to belong
to the core when agents have different and ambiguous knowledge, i.e. when priors
are not probability distributions but belief functions. It would be interesting
also to study when in the general case, the core of the expected game is not
part of the ex-ante core. Our result shows that this holds true with
probabilistic identical priors, but we suspect that this is not always true in
the general case. The same question can be asked for the ex-ante bargaining set
and the ex-ante kernel (for the ex-ante prenucleolus, we have proved that the
prenucleolus of the expected game can be seen as a certain grand contract, which
belongs to the ex-ante prenucleolus).
As for the bargaining set, it would be interesting to know if our notion of
strong bargaining set is the weakest possible to make true the coincidence of
the core and the bargaining set for convex games. Lastly, some applications to
voting games may be of interest, in particular in what concerns the bargaining
set.

\bmhead{Acknowledgements}
The first author acknowledges the financial support by the European Union HORIZON-WIDERA-2023-TALENTS-01 project No 101183743 (AGATE). The second author is member of the GNAMPA-INdAM research group.



\end{document}